\documentclass[a4paper,12pt]{amsart}
\usepackage[mathscr]{euscript}
\usepackage{amsfonts}
\usepackage{amssymb, amsmath}
\usepackage{array,booktabs,mathtools,tabularx}

\newcolumntype{L}{@{}X@{}}
\usepackage{url}
\usepackage{fixme}
\newtheorem{thm}{Theorem}
\newtheorem{cor}{Corollary}
\newtheorem{remark}[thm]{Remark}
\newtheorem{defn}{Definition}
\newcommand{\F}{\mathbb F_q}
\begin{document}
\title{Riemann-Roch Spaces and Linear Network Codes}
\author{Johan~P.~Hansen}
\address{Department of Mathematics, Aarhus University}
\email{matjph@imf.au.dk}
\date{\today}
\thanks{Part of this work was done while visiting Institut de Math\'ematiques de Luminy, MARSEILLE, France.
I thank for the hospitality shown to me. This work was supported by the Danish Council for Independent
Research, grant no. DFF-4002-00367. }
\begin{abstract}

We construct linear network codes utilizing algebraic curves over finite fields and certain associated Riemann-Roch spaces and present methods to obtain their parameters.

In particular we treat the Hermitian curve and the curves associated with the Suzuki and Ree groups all having the maximal number of points for curves of their respective genera.

Linear network coding transmits information in terms of a basis of a vector space and the information is received as a basis of a possibly altered vector space. Ralf Koetter and Frank R. Kschischang 
introduced a metric on the set of vector spaces and showed that a minimal distance decoder for this metric achieves correct decoding if the dimension of the intersection of  the transmitted and received vector space is sufficiently large.

The vector spaces in our construction have minimal distance
bounded from below in the above metric making them suitable for linear network coding. 
\end{abstract}
\maketitle
\tableofcontents

\subsection*{Notation}
\begin{itemize}
\item $\F$ is the finite field with $q$ elements of characteristic $p$.
\item $\mathbb F=\overline{\mathbb F_q}$ is an algebraic closure of $\F$.
\item $G(l,N)$ is the Grassmannian of $l$-dimensional $\mathbb F$-linear subspaces of $\mathbb F^N$ and $G(l,N)(\F)$ its $\F$-rational points, i.e. $l$-dimensional $\F$-linear subspaces of $\F^N$.
\end{itemize}

\section{Introduction}
\subsubsection*{Linear network coding}\label{network}
In linear network coding transmission is obtained by transmitting a number of packets into the network and each packet is regarded as a vector of length $N$ over a finite field $\F$. The packets travel the network through intermediate nodes, each forwarding $\F$-linear combinations of the packets it has available. Eventually the receiver tries to infer the originally transmitted packages from the packets that are received, see \cite{DBLP:citeseer_10.1.1.11.697} and \cite{Ho06arandom}.

Ralf Koetter and Frank R. Kschischang \cite{DBLP:journals/tit/KoetterK08}  endowed the Grassmannian $G(l,N)(\F)$ of $l$-dimensional $\F$-linear subspaces of $\F^N$ with the metric
\begin{eqnarray}\label{dist}
\mathrm{dist}(V_1,V_2):=&\dim_{\F}(V_1+V_2)-\dim_{\F}(V_1\cap V_2)=\\ &\dim(V_1)+\dim(V_2)-2 \dim(V_1\cap V_2)\ ,
\end{eqnarray}
where $V_1,V _2\in G(l,N)(\F)$.
\begin{defn}\label{setup}
A linear network code $\mathcal C \subseteq G(l,N)(\F)$ is a set of $l$-dimensional $\F$-linear subspaces of $\F^N$.

The size of the code $\mathcal C \subseteq G(l,N)(\F)$  is denoted by $\vert \mathcal C \vert$ and the minimal distance by
\begin{equation}
D(\mathcal C):= \min_{V_1,V_2 \in \mathcal C, V_1 \neq V_2} \mathrm{dist}(V_1,V_2)\ .
\end{equation}
The linear network code $\mathcal C$ is said to be of type $[N,l,\log_q \vert \mathcal C \vert, D(\mathcal C)]$. Its normalized weight is $\lambda = \frac{l}{N}$, its rate is $R= \frac{\log_q (\vert \mathcal C \vert)}{N l}$
and its normalized minimal distance  is $\delta = \frac{D(\mathcal C)}{2 l}$.
\end{defn}
Ralf Koetter and Frank R. Kschischang showed that a minimal distance decoder for this metric achieves correct decoding if the dimension of the intersection of  the transmitted and received vector-space is sufficiently large. Also they obtained  Hamming, Gilbert-Varshamov and Singleton coding bounds. 

\subsubsection*{Algebraic curves and Riemann-Roch spaces}

Let $X$ be an absolutely irreducible, projective algebraic curve of genus $g$ defined over the finite field $\F$. Let $X(\F)$ be the $\F$-rational points on $X$. 

To any subset $S \subseteq  X(\F)$  and any positive integer $k$, we associate the divisor $\sum_{P\in S} P \in \mathrm{Div}(X)$ and the Riemann-Roch spaces
\begin{equation}\label{code}
V= \mathrm{L}\Big(k\ \sum_{P\in S} P\Big) \subseteq \mathrm{L}\Big(k\ \sum_{P\in X(\F) } P\Big) = W\ .
\end{equation}
Certain collections of such subspaces $V \subseteq W$ will comprise our linear network code with ambient space $W$.

\subsubsection*{The general construction and applications in concrete cases}

In our construction, we obtain a subspace as in (\ref{code}) for each subset $S \subseteq  X(\F)$ of given size $s$. Using the Riemann-Roch theorem we are able to determine all the parameters of the resulting linear network codes depending on the number of $\F$-rational points on the curve $X$ and its genus $g$.

The potential of our construction relies on the ability to find curves with many $\F$-rational points, which is in fact possible. We recollect some of the theory of bounds on the number of $\F$-rational on curves in  \ref{sizeofcode}.

In \ref{DL} we discuss the Hermitian curve and the Deligne-Lutzig curves associated with the Suzuki and Ree groups all having the maximal number of points for curves of their genera.
\section{Construction of linear network codes from algebraic curves and Riemann-Roch spaces}

Let $X$ be a absolutely irreducible and projective algebraic curve of genus $g$ defined over the finite field $\F$. Let $X(\F)$ be the set of $\F$-rational points on $X$ and $n= \vert X(\F) \vert$ their number.

For a fixed positive integer $k$, let $k \cdot\sum_{P\in X(\F)} P \in \mathrm{Div}(X)$ be the Frobenius invariant divisor of degree $kn$ with support in all of the $\F$-rational points. 
The ambient vector space $W$ of the linear network codes is the associated Riemann-Roch space
\begin{equation}\label{W}
W=\mathrm{L}\Big(k\cdot\sum_{P\in X(\F) } P\Big)\ .
\end{equation}

From Riemann-Roch we have
\begin{equation}\label{NN}
\begin{cases}N=\dim W \geq  k n +1-g\\
N=\dim W = k n+1-g &\mathrm{for}\  k n \geq 2g-1 \end{cases}
\end{equation}

We refer to \cite{MR1042981} for the general theory of Riemann-Roch spaces. 
\begin{remark} Let $D \in \mathrm{Div}(X)$ be a Frobenius-invariant divisor on $X$, then the vector space $\mathrm{L}(D)$ has a basis of Frobenius-invariant vectors and
\begin{equation}\label{Frob}
\dim \mathrm{L}(D)=\dim_{\mathbb{F}}\mathrm{L}(D)= \dim_{\mathbb{\F}}\mathrm{L}(D)^{Fr}\ ,
\end{equation}
where $\mathrm{L}(D)^{Fr} \subseteq \mathrm{L}(D)$ denotes the subspace of Frobenius-invariant vectors in $\mathrm{L}(D)$.

As all our divisors are Frobenius-invariant and we will consistently use (\ref{Frob}).
\end{remark}

\begin{defn}\label{defLNC}
For a fixed positive integer $s$, the linear network code $\mathcal{C}_{k,s}$ of linear subspaces of $W$ in (\ref{W}) is constructed by associating to any subset $S \subseteq  X(\F)$  of size $s$, the Frobenius-invariant divisor ${k \cdot \sum_{P\in S} P}$ of degree $ks$ and its Riemann-Roch space $V=\mathrm{L}\Big(k\cdot \sum_{P\in S} P \Big)$.

Specifically
\begin{equation}
\mathcal{C}_{k,s}= \Big\{V=\mathrm{L}\Big(k\cdot \sum_{P\in S} P\Big)\subseteq  W\ \ \Big\vert \ S \subseteq X(\F),\  \vert S \vert = s \Big\}\ .
\end{equation}
\end{defn}

As for the dimension $l= \dim V$ of the linear subspaces ${V=\mathrm{L}\Big(k\cdot \sum_{P\in S} P\Big)}$ in the network code, the theorem of Riemann-Roch gives
\begin{equation}\label{ll}
\begin{cases}
l=\dim V = \dim \mathrm{L}\Big(k\cdot \sum_{P\in S} P\Big) \geq  ks +1-g\\
l=\dim V =\dim \mathrm{L}\Big(k\cdot \sum_{P\in S} P\Big)= ks+1-g &\mathrm{for}\  ks\geq 2g-1
\end{cases}
\end{equation}
with $S$ of size $s$, see \cite{MR1042981}.

As for the intersection of two linear subspaces 
$V_1=\mathrm{L}\Big(k\cdot \sum_{P\in S_1} P\Big)$ and $V_2=\mathrm{L}\Big(k\cdot \sum_{P\in S_2} P\Big)$ in the network code, we have from the definition of the Riemann-Roch spaces
\begin{equation}
V_1 \cap V_2 = \mathrm{L}\Big(k\cdot \sum_{P\in S_1} P\Big) \cap \mathrm{L}\Big(k\cdot \sum_{P\in S_2} P\Big)=\mathrm{L}\Big(k\cdot \sum_{P\in S_1\cap S_2} P\Big)
\end{equation}

If $S_1 \cap S_2 =\emptyset$, we have that $V_1 \cap V_2 = 0$ and $\dim V_1 \cap V_2 = 0$.

If $S_1 \cap S_2 \neq \emptyset$, the theorem of Riemann-Roch gives that
\begin{equation}\label{snit}
\begin{cases}
\dim V_1 \cap V_2 \geq  k \vert S_1 \cap S_2 \vert +1-g\\
\dim V_1 \cap V_2= k \vert S_1 \cap S_2 \vert+1-g\  \mathrm{ for }\  k \vert S_1 \cap S_2 \vert\geq 2g-1
\end{cases}
\end{equation}
as the divisor $k\cdot \sum_{P\in S_1\cap S_2}$ has degree $k \vert S_1 \cap S_2 \vert > 0$.

\begin{thm}\label{main} Let $X$ be an absolutely irreducible and projective algebraic curve of genus $g$ defined over the finite field $\F$. Let $X(\F)$ be the $\F$-rational points on $X$ and $n= \vert X(\F) \vert$ their number.

Let $\mathcal{C}_{k,s}$ be the linear network code of Definition \ref{defLNC}. 

Assume $k,s$ are positive integers with $ks\geq 2g-1$. 

The dimension $N$ of the ambient space $W$ is
\begin{equation}\label{N}
N= \dim W= kn+1-g\ .
\end{equation}
The dimension $l$ of the vector spaces $V \in \mathcal{C}_{k,s}$
is
\begin{equation}\label{l}
l=\dim V = ks+g-1 \ .
\end{equation}
The size of the code is
\begin{equation}\label{size}
\vert \mathcal{C}_{k,s} \vert = \binom{n}{s}\ .
\end{equation}
If $s=1$ the minimum distance of the code is
\begin{equation}\label{Ds=1}
D(\mathcal{C}_{k,s})=2(k+g-1)\ .
\end{equation}

If $s>1$, assume $k(s-1)\geq 2g-1$. 
The minimum distance of the code is
\begin{equation}\label{Ds>1}
D(\mathcal{C}_{k,s})=2k\ .
\end{equation}
\end{thm}
\begin{proof}
The claim (\ref{N}) follows from (\ref{NN}) and (\ref{l}) follows from (\ref{ll}). The claim in (\ref{size}) is obvious as there is a distinct vector space in the linear network code for each choice of $s$ points among the $n$ points in $\vert X(\F) \vert$.

Finally (\ref{Ds=1}) and (\ref{Ds>1}) follow from (\ref{snit}), as we obtain the minimal distance between two distinct vector spaces $V_1=\mathrm{L}\Big(k\cdot \sum_{P\in S_1} P\Big)$ and $V_2=\mathrm{L}\Big(k\cdot \sum_{P\in S_2} P\Big)$ in the network code when their intersection  has maximal dimension.

I case $s=1$ the intersection always has dimension 0. 
From the definition of the metric in (\ref{dist}) and (\ref{l}), we conclude
\begin{equation}
\mathrm{dist}(V_1,V_2)=2 (k+1-g)\ .
\end{equation}

In case $s>1$ the maximal dimension of the intersection is obtained when $\vert S_1 \cap S_2\vert =s-1$ and under the assumption $k(s-1)\geq 2g-1$, we have
\begin{equation}
\dim V_1 \cap V_2= k \vert S_1 \cap S_2 \vert+1-g =k(s-1)+1-g\ .
\end{equation}
From the definition of the metric in (\ref{dist}) and (\ref{l}), we conclude
\begin{equation}
\mathrm{dist}(V_1,V_2)=2 (ks+1-g)-2(k(s-1)+1-g)=2k \ .
\end{equation}
\end{proof}

\begin{cor}\label{cor} Under the assumptions of the theorem and in the notation of Definition \ref{setup} the normalized weight of the code $\mathcal{C}_{k,s}$ is 
\begin{equation}
\lambda= \frac{ks+1-g}{kn+1-g} \ .
\end{equation}

The rate of the code is 
\begin{equation}
R= \frac{\log_q \big(\binom{n}{s}\big)}{(kn+1-g)(ks+1-g)}\ .
\end{equation}

The normalized minimal distance $\delta$ of the code satisfies
\begin{equation}\label{delta}
\delta \geq \frac{2g-1}{(s+1)g-1} \ .
\end{equation}
\end{cor}
\begin{proof}
Only the claim on the normalized minimal distance is non-trivial. 

In case $s=1$, two distinct vector spaces in the linear network code has trivial intersection and the normalized minimal distance $\delta$ is 1.

In case $s>1$, we get from the theorem that
\begin{equation}
\delta = \frac{2k}{2(ks+1-g)}=\frac{1}{s+\frac{1-g}{k}} \ .
\end{equation}
By assumption $k \geq \frac{2g-1}{s-1}$ and (\ref{delta}) follows.
\end{proof}
\subsection{Sizes of the codes and the number of rational points on the curves}\label{sizeofcode}

Let $\F$ be the field with $q$ elements, and let $X$ be a projective and absolutely irreducible algebraic curve of genus $g$ defined over $\F$. 

In order to produce linear network codes of large size, curves with a larger number $\vert X(\F) \vert$ of $\F$-rational points are needed.

The Hasse-Weil bound asserts
\begin{equation}\label{Hasse-Weil}
1+q-2g \sqrt{q}\leq \vert X(\F) \vert \leq 1+g + 2g \sqrt{q} \ .
\end{equation}

For a given genus $g$, the Hasse-Weil bound (\ref{Hasse-Weil}) can often be improved, in particular when the genus $g$ is large compared to the field size $q$.

Let $N_q(g)$  be the maximum number
of $\F$-rational points on any curve over $\F$ of genus $g$.

Then
\begin{equation}\label{lighed}
\limsup_{g \rightarrow \infty}\frac{N_q(g)}{g} = \sqrt{q}-1
\end{equation}
for square cardinalities $q$.

Drinfeld and Vladut [3] derived the bound
\begin{equation}
\limsup_{g \rightarrow \infty}\frac{N_q(g)}{g}\leq \sqrt{q}-1
\end{equation}
for fixed $q$.

Ihara proved in \cite{MR656048} that
\begin{equation}
\limsup_{g \rightarrow \infty}\frac{N_q(g)}{g} = \sqrt{q}-1
\end{equation}
for square cardinalities $q$.
This was again proved by Tsfasman, Vladut and Zink in
\cite{MR705893}.

Garcia and Stichtenoth wrote down explicit
towers of field extensions in \cite{MR1345289}
realizing the equality in (\ref{lighed}), see also \cite{MR2278034}, \cite{MR2483217} and \cite{MR2674211}.
For the general theory of function fields, see \cite{MR2464941}. 

Here we will not study the linear network codes constructed
from the towers of Garcia and Stichtenoth, but proceed to present codes from Deligne-Lusztig curves all having the maximal number of $\F$-rational points allowed for their genera.

\subsection{Linear network codes from Deligne-Luztig Curves}\label{DL}

Linear network codes can be constructed from Riemann-Roch spaces on Deligne-Lusztig curves associated to a connected reductive algebraic group $G$ defined over a finite field $\F$. These curves was originally introduced in \cite{MR0393266}.

The Deligne-Lusztig curves used in the construction of the codes have in some cases many $\F$-rational points - in fact the maximal number in relation to their genera as determined by the “explicit formulas” of Weil.

The relevant groups for Deligne-Lusztig curves are groups of $\mathbb{F}_q$-rank 1. There are only four such groups: $A_1(q)$, $^2A_2(q^2)$, $^2B_2(q^2=2^{2k+1})$, $^2G_2(q^2=3^{2k+1})$.

The corresponding Deligne-Lusztig curves are smooth, projective curves over $\F$.

In \cite{MR1186416} the genera of and the number of rational points on the corresponding curves are determined in all 4 cases:
\begin{itemize}
    \item[i)] $A_1$: $X=\mathbb{P}^1$. It has genus 0 and $1+q$ over $\F$
    \item[ii)] $^2A_2$: The Fermat curve $X: x^{q+1}+y^{q+1}=z^{q+1}$ of degree $q+1$. It has genus $q(q-1)/2$ and $1+q^3$ points over $\mathbb{F}_{q^2}$.
    \item[iii)] $^2B_2$: The Deligne-Lusztig curve of Suzuki type. It has genus $q(q^2-1)/\sqrt{2}$ and $1+q^4$ points over $\mathbb{F}_{q^2}$.
    \item[iv)] $^2G_2$: The Deligne-Lusztig curve of Ree type. It has genus $\sqrt{3}q(q^4-1)/2+q^2(q^2-1)/2$ and has $1+q^6$ points over $\mathbb{F}_{q^2}$.
\end{itemize}

See also \cite{MR1325513} for the curves of Suzuki type, where bases for the vector spaces $\mathrm{L}(P)$ are determined and 
\cite{MR1225959} for the curves of Ree type.

The parameters of the resulting linear network codes are obtained by substituting the values of $g$ and $\vert X(\F) \vert$ in the formulas of Theorem \ref{main} and Corollary \ref{cor}.
\bibliography{netRR}{}
\bibliographystyle{plain}
\end{document}